\documentclass[11pt,letterpaper]{article}

\usepackage[margin=1in]{geometry}
\usepackage[T1]{fontenc}
\usepackage{lmodern}
\usepackage{microtype}

\usepackage{amsmath,amsfonts,amssymb}
\usepackage{amsthm}
\usepackage{mathtools}
\usepackage{physics}
\usepackage{orcidlink}
\usepackage{array}
\usepackage{multirow}
\usepackage{graphicx}
\usepackage{float}
\usepackage{textcomp}
\usepackage{url}
\usepackage{cite}
\usepackage{authblk}
\usepackage{mathdots}
\usepackage[font=normalsize]{caption}
\newenvironment{IEEEkeywords}%
{\par\medskip\noindent\textbf{Keywords: }}%
{\par}

\newcolumntype{C}[1]{>{\centering\arraybackslash}m{#1}}

\newcommand{\A}{\mathcal{A}}
\newcommand{\B}{\mathcal{B}}

\newcommand{\Id}{{\rm 1\hspace{-0.9mm}l}}

\newcommand{\C}{\ensuremath{\mathbb{C}}}

\newcommand{\R}{\ensuremath{\mathbb{R}}}

\newcommand{\proj}[1]{\ensuremath{\ketbra{#1}{#1}}}

\newtheorem{theorem}{Theorem}
\newtheorem{corollary}[theorem]{Corollary}

\newtheorem{conjecture}[theorem]{Conjecture}

\usepackage{enumitem}

\newtheorem{lemma}[theorem]{Lemma}
\theoremstyle{definition}
\newtheorem{definition}[theorem]{Definition}

\theoremstyle{remark}

\numberwithin{equation}{section}

\newcommand{\Mn}[1]{M_{#1}(\C)}
\newcommand{\Sn}[1]{\Omega_{#1}}
\DeclareMathOperator{\sw}{sw} 
\DeclareMathOperator{\swc}{\overline{sw}}

\newcommand{\sws}{\sw^{*}}

\DeclareMathOperator{\conv}{conv}
\DeclareMathOperator{\diag}{diag}
\DeclareMathOperator{\vspan}{span}

\newcommand{\psd}{\succeq 0}

\begin{document}

\title{Parallel quantum channel discrimination and numerical ranges in tensor product subspaces}

\author[1,2]{Adam B\'ilek\,\orcidlink{0000-0001-8380-5338}\thanks{Corresponding author: Adam B\'ilek (e-mail: adam.bilek@vsb.cz).}}
\author[2]{Paulina Lewandowska\,\orcidlink{0000-0003-1564-7782}}
\author[2]{Ryszard Kukulski\,\orcidlink{0000-0002-9171-1734}}
\affil[1]{Department of Applied Mathematics, Faculty of Electrical Engineering and Computer Science, VSB--Technical University of Ostrava, 17.~listopadu 2172/15, 708 33 Ostrava, Czech Republic}
\affil[2]{IT4Innovations, VSB--Technical University of Ostrava, 17.~listopadu 2172/15, 708 33 Ostrava, Czech Republic}
\date{}
\maketitle

\begin{abstract}
Quantum channel discrimination plays a crucial role in quantum information theory. Of particular interest is the case in which the channels can be discriminated perfectly.
In this work, we focus on the perfect quantum channel discrimination task in a parallel scheme. We develop
an SDP formulation combined with a bisection procedure to compute a quantum state for perfect discrimination in time linear in the number of copies. In addition, we obtain the minimal number of copies of quantum channels to achieve perfect discrimination.  Thanks to that, we  settle in the affirmative Conjecture~1 of Duan, Guo, Li and Li
\cite{duan2016parallel},
which characterizes the number of parallel uses needed to discriminate perfectly a distinguished family of operator subspaces.  All our results are obtained using the notion and basic properties of the numerical range. In particular, the key fact that we prove and use is that the minimal angle of the numerical range of a tensor product of matrix subspaces equals the sum of the minimal angles of the numerical ranges of the individual subspaces.
\end{abstract}


\section{Introduction}\label{sec:intro}

One of the fundamental and well-studied problems in quantum information theory is the discrimination (or distinguishability) of quantum objects. In this problem, we are given a quantum object or many identical copies of it, which is known to be in one of two possible states, and we are given the classical descriptions of both states. The goal is to determine which description corresponds to the actual state of our object by interacting with it, that is, by designing a quantum experiment: a discrimination scheme.

The study of this problem began with the seminal works of Helstrom and Holevo~\cite{helstrom1969quantum, holevo1973bounds}.
Later, extensive studies were conducted for the discrimination of unitary operators \cite{acin2001statistical, duan2007entanglement, duan2008local}, for quantum measurements \cite{puchala2018strategies, krawiec2020discrimination}, or  process matrices \cite{lewandowska2023strategies}. 
This has also been extended to the case of discrimination between more than two channels \cite{chiribella2013identification}, as well as 
 quantum measurements    \cite{krawiec2020discrimination, puchala2021multiple}. 
In the work \cite{duan2009perfect} the authors formulated the necessary and sufficient conditions under which quantum channels can be perfectly discriminated. 

It turned out that the perfect distinguishability 
of the quantum channels has been connected to their Kraus decomposition \cite{watrous2018theory}. If 
$\Phi = (A_i)_i$ and $\Psi = (B_j)_j$ are given by their Kraus operators $A_i$ and $B_j$, respectively, then the subspace $\mathrm{span}\{B_j^\dagger A_i\}_{j,i}$ contains the information whenever  $p_{\text{succ}} = 1$. This means that there is a quantum 
state $\proj{\psi} $ such that $(\Phi  \otimes 
\Id)(\proj{\psi})$ and $(\Psi \otimes 
\Id) (\proj{\psi})$ are orthogonal, which is equivalent to 
$\bra{\psi}(B_j^\dagger A_i \otimes \Id)\ket{\psi} = 0$ for any $i,j$. 
Let $\ket{\psi}$ be a purification of a quantum state $\rho  $. Then, we 
arrive at the condition
\begin{equation}
p_{\text{succ}} = 1 \iff \exists_{\rho } \forall_{i,j} 
\tr(B_j^\dagger A_i \rho) = 0.
\end{equation}

Further works investigated the general scheme of multiple-shot discrimination, called adaptive strategies \cite{harrow2010adaptive, cope2018adaptive, pirandola2017ultimate}. The multiple-shot discrimination task has also entered the experimental phase in various scenarios
e.g. in \cite{liu2019distinguishing, debry2023experimental, bilek2026experimental}. 
The special case of that,  parallel  schemes were studied in \cite{duan2016parallel}.  
In this scheme, we discriminate between $N$ copies of $\Phi$ and $\Psi$, i.e.,
between $\Phi^{\otimes N}$ and $\Psi^{\otimes N}$. To analyse this case, we
associate the problem with the subspace
$\vspan\{B^\dagger_{j_1}A_{i_1} \otimes \cdots \otimes
B^\dagger_{j_N}A_{i_N}\}_{\mathbf{i},\mathbf{j}}$. The Authors~\cite{duan2016parallel} showed that the perfect distinguishability of two quantum  operations by a parallel scheme depends only on an operator subspace generated from their Choi--Kraus operators, and that conversely every operator subspace arises in this way from some pair of operations. For subspaces that are one-dimensional or Hermitian they proved that parallel distinguishability is equivalent to the absence of a positive definite operator in the subspace, and gave the corresponding optimal protocol; in general, however, this condition is only necessary. They also consider a family of operator subspaces that is neither one-dimensional nor Hermitian, and therefore falls outside the scope of their theorems. For this family, they formulate the following conjecture, whose validity they leave open.
\begin{conjecture}[Conjecture 1. in \cite{duan2016parallel}]\label{conj}
For operator space $S_\beta = \vspan\{\proj{0}+e^{i \beta} \proj{1}, \proj{1}+e^{i \beta} \proj{2}\}$ with $\beta \in [\pi/2, \pi]$, there is a
density operator in the orthogonal complement of $(S_\beta)^{\otimes N}$ if and only if $\beta \in [\pi/2 + \pi/(2N), \pi]$.
\end{conjecture}
In \cite{li2019note}, the Authors have affirmed the necessity part of the conjecture and established the sufficiency of the conjecture
for $N \le 10$ by presenting explicit non-trivial nonnegative solutions for the linear system.

In this work, we establish necessary and sufficient conditions for perfect parallel discrimination of quantum channels, together with a constructive method for finding the discriminator. Moreover, we compute the minimal number of copies $N$ required for perfect discrimination. Our results generalize those obtained in \cite{acin2001statistical} for unitary channels, as well as the later results for von Neumann measurements given in Corollary 1 of \cite{puchala2021multiple}. The construction relies on semidefinite programming \cite{vandenberghe1996sdp} and on the notion of numerical range \cite{numericalshadow}, whose properties we investigate in detail. In particular, we prove that the minimal angle of the numerical range of the tensor product of the matrix subspaces is equal to the sum of the minimal angles of the numerical ranges of those subspaces. It should also be noted that our results can be linked to the study of the product numerical range \cite{Puchala2011,Gawron2010} and the joint numerical range \cite{li2020joint,plaumann2021kippenhahn,keeler1997numerical,szymanski2018classification}.

This paper is organized as follows. Section \ref{sec:prelim}   presents the necessary mathematical preliminaries. Section \ref{sec:main} demonstrates the main results of the work.  Finally, concluding remarks are presented in the final Section \ref{sec:concl}.
In the appendices we provide the technical proof of the theorems presented in the work.

\section{Mathematical preliminaries}
\label{sec:prelim}

Throughout, $\Mn{n}$ denotes the algebra of complex matrices $n\times n$ and 
$\Id_{n}\in\Mn{n}$ the identity matrix. We will be using the notation of linear subspaces $\A, \B \subset \Mn{n}$ as well as the notation of tensor product of linear subspaces $\A \otimes \B = \vspan\{A \otimes B: A\in\A, B\in \B\}$. We say that the matrix $A$ is Hermitian if  $A=A^{\dagger}$ where $A^{\dagger}$  denotes the conjugate transpose of $A$. We write $A\psd$ to indicate that $A$ is
positive semidefinite.  Every
$A\in\Mn{n}$ decomposes as $A=\Re( A)+i\,\Im (A)$, where
$
  \Re (A)=\tfrac12\bigl(A+A^{\dagger}\bigr) $ and 
$  \Im (A)=\tfrac1{2i}\bigl(A-A^{\dagger}\bigr) $
are both Hermitian. 
  A matrix $\rho\in\Mn{n}$ is a \emph{quantum state} or \emph{density matrix}
 if $\rho\psd$ and $\tr\rho=1$. The set of quantum states in $\Mn{n}$ is
  denoted by $\Sn{n}$.
Finally, we need a one-parameter relaxation of the state space, in which
the two halves of a state are allowed to carry opposite phases.
Let $\alpha\in[0,\pi)$, we define $\Sn{n}(\alpha)$ as the set of all matrices
  $\rho\in\Mn{n}$ of the form
  \begin{equation}\label{def-states}
  \Sn{n}(\alpha) = \left\{ \rho \in \Mn{n}: \rho= e^{i\frac{\alpha}{2}}P + e^{-i\frac{\alpha}{2}}Q, \,\, P,Q\psd, \,\, \tr(P)+\tr(Q)=1 \right\}.
  \end{equation}
Observe that the definition above extends the notion of quantum states as $\Sn{n}=\Sn{n}(0)$.

We now recall the central tool for the discrimination task, the numerical range.   The \emph{numerical range} of $X\in\Mn{n}$ is
  $W(X) = \bigl\{\bra{v} X \ket{v}: \ket{v}\in\C^{n}, \braket{v}{v}=1 \bigr\}$.
  Alternatively, it is defined as $ W(X) = \bigl\{\tr(\rho X): \rho \in \Sn{n}\bigr\}$. In particular $0\in W(X)$ precisely when $\tr(\rho X)=0$ for some state $\rho$,
and this reformulation is the one that we use throughout. By the
Toeplitz--Hausdorff theorem, $W(X)$ is a compact convex subset of $\C$, and
$W(X)\subseteq\R$ exactly when $X=X^{\dagger}$.

For $X\in\Mn{n}$, we put $\Gamma_X = \{\theta: \Re(e^{i\theta}X)\psd, \theta \in [-\pi,\pi)\}$
and $|\Gamma_X|$ as volume of $\Gamma_X$. Additionally, we define
   $\sw(X) \in [0,\pi]$ as 
    \begin{equation}\label{eq:sw-def}
    \sw(X)\;=\;\begin{cases} \pi, & 0\in W(X), \\ \pi - |\Gamma_X|, & 0\notin W(X).
    \end{cases}
  \end{equation}
Additionally, for a given non-empty set  $\mathcal{M} \subset \Mn{n}$ we define 
\begin{equation}
    \sws(\mathcal{M}) =\inf\{\sw(X):  X\in\mathcal{M}\}.
\end{equation}

\section{Main results}\label{sec:main}
In this section, we present the main results of this paper. First, we define an SDP of our interest. Due to this SDP, we will be able to construct necessary and sufficient conditions for perfect discrimination of quantum channels. 

Let $\A \subset \Mn{n}$ be a  linear subspace, $(A_l)_{l=1}^k$ be a basis of $\A$ and $\alpha \in [0,\pi)$. We define $D_\alpha(A_1,\ldots,A_k)$ as the output of the following optimization program:
    \begin{equation}\label{sdp-main}
        \begin{split}
        & \textbf{\underline{SDP formulation for computing $D_\alpha(A_1,\ldots,A_k)$}}\\
            \text{minimize:}\quad &
            \sum_{l=1}^k \left| \tr(\rho A_l)\right|
            \\[2mm]
            \text{subject to:}\quad & \rho \in \Sn{n}(\alpha)
        \end{split}
    \end{equation}
In Appendix~\ref{app-sdp-spec} we provide a ready to implement version of the Program~\ref{sdp-main}. Now, the main result of our work can be summarized as the following theorem.
\begin{theorem}\label{thm-tool}
    Let $\A \subset \Mn{n}$ be a linear subspace, and $\alpha \in [0,\pi)$. The following conditions are equivalent:
    \begin{enumerate}
        \item $\sws(\A) + \alpha \ge \pi$,
        \item $D_\alpha(A_1,\ldots,A_k) = 0$ for any basis $(A_l)_{l=1}^k$ of $\A$,
        \item There is $\rho \in \Sn{n}(\alpha)$ such that $\tr(\rho A)=0$ for all $A \in \A$.
    \end{enumerate}
\end{theorem}
The proof of Theorem~\ref{thm-tool} can be found in Appendix~\ref{app-proof-thm-tool}. Thanks to this theorem for any linear subspace $\A \subset \Mn{n}$, we are able to efficiently estimate $\sws(\A) \in [0,\pi]$ based on the bisection method:
\begin{itemize}
    \item if $D_\alpha(A_1,\ldots,A_k) = 0$ for some $\alpha$, then $\sws(\A) \ge \pi - \alpha$,
    \item if  $D_\alpha(A_1,\ldots,A_k) > 0$ for some $\alpha$, then $\sws(\A) < \pi - \alpha$.
\end{itemize}
Eventually, we can reach $\sws(\A)$ with arbitrary precision. Now, let us define
\begin{equation}
    \alpha_\A = \pi - \sws(\A).
\end{equation} 
Moreover, as long as $\sws(\A) > 0$, we also define $\rho_\A$ of the form 
\begin{equation}
    \rho_\A = e^{i \frac{\alpha_\A}{2}} P_\A +  e^{-i \frac{\alpha_\A}{2}} Q_\A \in \Sn{n}(\alpha_\A) \text{ such that } \tr(\rho_\A A)=0 \text{ for any } A \in \A.
\end{equation}
Otherwise, if $\sws(\A) = 0$, we define $\rho_\A = \Id_n/n$.

In the following theorem, we present the properties of the function
$\sw^*(\cdot)$ and, as a consequence, derive a formula for the discriminator. 

\begin{theorem}\label{thm-tensor}
    Let $\A \subset \Mn{n}$ and $\B \subset \Mn{m}$ be any linear subspaces. It holds that
    \begin{equation}
        \sws \left(\A \otimes \B \right) = \min\left\{ \sws(\A) + \sws(\B), \pi \right\}.
    \end{equation}
    Moreover, we have the following construction of $\rho_{\A \otimes \B} \in \Sn{nm}(\alpha_{\A \otimes \B})$:
    \begin{itemize}
        \item If $\sws(\A) = 0$ or $\sws(\B) = 0$ we take
        \begin{equation}
            \rho_{\A \otimes \B} = \rho_\A \otimes \rho_\B.
        \end{equation}
        \item If $\sws(\A) > 0$, $\sws(\B) > 0$ and $\sws(\A) + \sws(\B) < \pi$ we take
        \begin{equation}
            \rho_{\A \otimes \B} \propto e^{i\frac{\alpha_\A + \alpha_\B - \pi}{2}} P_\A \otimes P_\B + e^{-i\frac{\alpha_\A + \alpha_\B - \pi}{2}} Q_\A \otimes Q_\B.
        \end{equation}
        \item If $\sws(\A) > 0$, $\sws(\B) > 0$ and $\sws(\A) + \sws(\B) \ge \pi$ we take
        \begin{equation}
            \rho_{\A \otimes \B} \propto \rho_\A \otimes \rho_\B^\dagger + \rho_\A^\dagger \otimes \rho_\B.
        \end{equation}
    \end{itemize}
\end{theorem}
The proof of Theorem~\ref{thm-tensor} can be seen  in Appendix~\ref{app-proof-thm-tensor}. The results of this theorem can be extended naturally to tensor product of $N$ subspaces.  
\begin{corollary}\label{cor-many}
  Let $\A_i \subset \Mn{n_i}$ be linear subspaces for $i=1,\ldots,N$. Then, it holds
  \begin{equation}
      \sws \left(\bigotimes_{i=1}^N \A_i \right) = \min\left\{ \sum_{i=1}^N \sws(\A_i), \pi \right\}.
  \end{equation}   
  In particular, $\sw \left(\bigotimes_{i=1}^N A_i \right) = \min\left\{ \sum_{i=1}^N \sw(A_i), \pi \right\}$ for any sequence of matrices $A_i$. 
\end{corollary}
The construction of $\rho_{\A_1 \otimes \cdots \otimes \A_N}$ can be delivered iteratively from Theorem~\ref{thm-tensor}. One of the main consequences of Corollary~\ref{cor-many} is a positive answer to Conjecture 1~\cite{duan2016parallel}.
\begin{corollary}\label{cor-conj}
    Conjecture~\ref{conj} is true.
\end{corollary}
The proof of Corollary~\ref{cor-conj} can be seen in Appendix~\ref{app-proof-cor-conj}. The result of Theorem~\ref{thm-tensor} can also be expressed in the language
of numerical ranges in the case $\sws(\A\otimes \B) = \pi$ presented in Corollary~\ref{cor-main}. 
\begin{corollary}\label{cor-main}
Let $\A\subset\Mn{n}$ and $\B\subset\Mn{m}$ be linear subspaces. Then
  the following are equivalent:
  \begin{enumerate}
    \item $0\in W(A\otimes B)$ for all $A\in\A$ and
          $B\in\B$.
    \item $0\in W(X)$ for all
          $X\in \A\otimes\B$.
  \end{enumerate}
\end{corollary}
Finally, we deliver the minimal number of  copies of quantum channels to achieve perfect discrimination. 
\begin{theorem}\label{thm-discr}
    Let $\Phi,\Psi: \Mn{d} \to \Mn{s}$ be the quantum channels given in the Kraus representation as  $\Phi = (A_k)_{k=1}^{r_{\Phi}}$ and $\Psi = (B_l)_{l=1}^{r_{\Psi}}$. Let us define the linear subspace $\A = \vspan\{B^\dagger_l A_k\}_{k,l}$. Then, the quantum channels $\Phi$ and $\Psi$ are perfectly discriminable in the parallel scheme if and only if $\sws(\A) > 0$. Moreover, the minimal number of copies $N_0$ required for perfect discrimination satisfies
    \begin{equation}
        N_0 = \left\lceil \frac{\pi}{\sws(\A)} \right\rceil.
    \end{equation}
\end{theorem}
The proof of Theorem~\ref{thm-discr} is presented in Appendix~\ref{app-proof-thm-discr}.

\section{Conclusion and discussion}\label{sec:concl}

In this work, we have provided a necessary and sufficient condition for the perfect 
discrimination of general quantum channels in the parallel scheme.
We have also given an explicit formula for the discriminator and determined the 
number of copies required for perfect discrimination. As a consequence, we have resolved
Conjecture~1 of~\cite{duan2016parallel}. Along the way, we have established several
properties of the numerical range that may be of independent interest. In
particular, we have shown that the minimal angle of the numerical range of a
tensor product of matrix subspaces is equal to the sum of the minimal angles of
the numerical ranges of the individual subspaces. Our results on parallel
discrimination also naturally extend to the discrimination of product
channels.

Recall that the optimal probability of successfully discriminating two
equiprobable quantum channels $\Phi$ and $\Psi$ is directly related to the
diamond norm through $p_{\text{succ}} = \frac{1}{2} + \frac{1}{4}\|\Phi - \Psi\|_\diamond$.
Hence, our results
provide an efficient method for deciding whether the diamond norm of the
difference of two product channels attains its maximal value. In the parallel
scheme, however, the behavior of the success probability when perfect
discrimination is not possible remains an interesting open problem.

From a broader perspective, the case $p_{\text{succ}}= 1$ is characterized by the angle of
the numerical range, namely by the condition $\sw(A) = \pi$ for certain
matrices $A$ associated with the Kraus operators of the channels.
Equivalently, this condition can be expressed in terms of the distance of
$W(A)$ from the origin as $\mathrm{dist}(0, W(A)) = 0$. For the case
$p_{\text{succ}} < 1$, one of these two characteristics, the angle or the distance,
may turn out to be the key quantity. If the angle is the relevant one, the
computation of the diamond norm might be simplified by the additivity
property $\sw(A \otimes B) = \sw(A) + \sw(B)$. However, if the distance is
the relevant quantity, the problem becomes considerably harder since, in 
general, the distance of $0$ from $W(A \otimes B)$ depends on global
properties of $A \otimes B$ rather than only on those of the individual
factors. Determining which of these two quantities governs the case $p_{\text{succ}} < 1$
is, in our view, a key step towards a complete description of parallel
channel discrimination. 

\section*{Acknowledgments}

The width $\sw$ defined in Eq.\eqref{eq:sw-def}, the additivity of widths under
tensor products that underlies Theorem~\ref{thm-tensor}, and a first outline
of the proof of Corollary~\ref{cor-main} were obtained in interactive sessions
with Claude Fable~5 (Anthropic), directed by A.B., after A.B.\ had proposed
the reformulation that made these tools applicable. P.L.\ and R.K.\ then
developed the shorter proofs presented here and derived semidefinite program that allows efficient
estimation of $\sws$ and a search for the witness state. R.K.\ proposed the
original research idea, and P.L.\ wrote the manuscript. The authors used AI
tools to improve the clarity and grammar of the text, they independently verified every
step of the manuscript, and they take full responsibility for the content of this publication.

A.B. and P.L. are supported by the Ministry of Education, Youth and Sports of the Czech Republic through the e-INFRA CZ (ID:90254).
R.K. is supported by the European Union under the Quantum error correction codes enhanced by reinforcement learning dedicated for the Ising model-based optimization, contract nr. 01906/2025/RRC via the Operational Programme Just Transition and Moravian-Silesian Region. 

\appendix
\section{Native SDP~\ref{sdp-main}}\label{app-sdp-spec}
In the following formulation, we use the Schur complement lemma (Lemma 3.18 in \cite{watrous2018theory}):
\begin{equation}\label{sdp-spec}
        \begin{split}
        & \textbf{\underline{SDP formulation for computing $D_\alpha(A_1,\ldots,A_k)$}}\\
            \text{minimize:}\quad &
            \sum_{l=1}^k c_l
            \\[2mm]
            \text{subject to:}\quad & \begin{pmatrix}
                c_l & \tr(\rho A_l) \\
                \overline{\tr(\rho A_l)} & c_l
            \end{pmatrix} \psd \quad l=1,\ldots,k;\\
            & \rho = e^{i \frac{\alpha}{2}} P + e^{-i \frac{\alpha}{2}} Q;\\
            &P,Q \psd;\\
            &\tr(P)+\tr(Q) = 1;\\[2mm]
            \text{variables:}\quad & c_l \in \R \quad l=1,\ldots,k;\\
            & P,Q \in \Mn{n}.
        \end{split}
    \end{equation}

\section{The proof of Theorem~\ref{thm-tool}}\label{app-proof-thm-tool}
We introduce the definitions and technical lemmas required to prove the theorem. First, we observe that for any $A \neq 0$ we have $|\Gamma_A| \le \pi$.
\begin{definition}
For any non-zero matrix $0 \neq A \in\Mn{n} $ we define $\swc(A) \in [0,\pi]$ as $\swc(A) = \pi - |\Gamma_A|$.
\end{definition}
We have some basic facts, which we leave without proof
\begin{lemma}\label{lemma-basic}
Let $A \in \Mn{n}$. The following is true:
    \begin{enumerate}
    \item $\sw(A) = \swc(A)$ whenever $0 \notin W(A)$ ($\sw(A) < \pi$),
    \item $\sw(A) = \sw(e^{i \alpha} A)$ for any $\alpha \in \R$,
    \item $\sw(A) = \sw(GAG^\dagger)$ for any full rank matrix $G \in \Mn{n}$,
    \item $\Re W(A) \ge 0$ if and only if $\Re(A) \psd$,
    \item If $A \neq 0$ and  $\swc(A) < \pi$, then $\tr(A) \neq 0$.
    \end{enumerate}
\end{lemma}

\begin{lemma}\label{lemma-no-zero}
Let $0 \neq A \in\Mn{n}$ be a matrix that $\swc(A) < \pi$. Then we can decompose $A = GDG^\dagger$, where $G$ is a nonsingular matrix and $D$ is a diagonal matrix with $D_{k,k}: |D_{k,k}| \in \{0, 1\}$ for $k=1,\ldots,n$.
\end{lemma}
\begin{proof}
   Based on Theorem 1.1 in \cite{horn2006canonical} there is a nonsingular matrix $G$ such that $GAG^\dagger = M_1 \oplus \ldots \oplus M_k$, where $M_l$ are matrices of three Types: 
   \begin{enumerate}
       \item $J_m(0)$, where $J_m(\lambda) = \begin{pmatrix}
           \lambda & 1 & & 0\\
            & \lambda & \ddots & \\
            & & \ddots & 1\\
            0 & & & \lambda
       \end{pmatrix}$ and $J_1(\lambda) = \lambda$.
       \item $\lambda \Delta_m$, where $|\lambda|=1$ and $\Delta_m = \begin{pmatrix}
           0 &  & & 1\\
            &  & \iddots & i \\
            & 1& \iddots & \\
            1 & i & & 0
       \end{pmatrix}$ and $\Delta_1 = 1$.
       \item $H_{2m}(\mu)$, where $|\mu| > 1$ and $H_{2m}(\mu) = \begin{pmatrix}
           0 & \Id_m\\
           J_m(\mu) & 0
       \end{pmatrix}$.
   \end{enumerate} 
 By assumption we have $\swc(M_l) < \pi$ for $M_l \neq 0$. Therefore, in our decomposition matrices $M_l$ are of the following types: $J_1(0)$ and $\lambda \Delta_1$. 
\end{proof}

\begin{lemma}\label{lem-prop-states}
    For any $\alpha \in [0,\pi)$ the set $\Sn{n}(\alpha)$ is non-empty, convex and compact. In particular, if $\rho \in \Sn{n}(\alpha)$, then $\swc(\rho) \le \alpha$.
\end{lemma}
\begin{proof}
It is easy to see that $\Sn{n}(\alpha)$ is convex. For any $\rho \in \Sn{n}(\alpha)$ we have $\|\rho\|_1 \le 1$, so $\Sn{n}(\alpha)$ is bounded. Moreover, if $\rho = \lim_{p \to \infty} \rho_p$, where $\rho_p \in \Sn{n}(\alpha)$, then $\rho_p = e^{i\frac{\alpha}{2}}P_p + e^{-i\frac{\alpha}{2}}Q_p$. The sequences $(P_p)_P, (Q_p)_p$ belong to the compact set $\{M \psd, \tr(M) \le 1\}$, therefore, we can choose appropriate subsequences $P_{p_l} \to P$ and $Q_{p_l} \to Q$, where $P,Q \psd$ and $\tr(P+Q)=1$. Therefore, $\rho = \lim_{l \to \infty} \rho_{p_l} = e^{i\frac{\alpha}{2}}P + e^{-i\frac{\alpha}{2}}Q \in \Sn{n}(\alpha)$. We have shown that the set $\Sn{n}(\alpha)$ is compact. Finally, for $\rho \in \Sn{n}(\alpha)$ we get $\rho \neq 0$ and $[-(\pi/2-\alpha/2),\pi/2-\alpha/2] \subset \Gamma_\rho$. That implies $\swc(\rho) \le \alpha$.
\end{proof}

\begin{lemma}\label{lemma4}
   Let us define $A \in \Mn{n}$. For any $0 \neq X \in \Mn{n}$ such that $\tr(X A) = 0$ we have $\swc(X) \ge \pi - \sw(A)$.
    Moreover, if $\sw(A) > 0$, then for each $\alpha \in [\pi - \sw(A), \pi)$ there is $\rho \in \Sn{n}(\alpha)$ such that $\tr(\rho A) = 0$.
\end{lemma}
\begin{proof}
    If $\sw(A) < \pi$, then $A=GDG^\dagger$ from Lemma~\ref{lemma-no-zero}. Then, we have $\sum_k D_{k,k} \bra{k} G^\dagger X G \ket{k} = 0$ with $|D_{k,k}|=1$. Let us assume indirectly $\swc(X) < \pi - \sw(A)$. Then, we have $\bra{k} G^\dagger X G \ket{k}=0$ for each $k$, which implies $\tr(G^\dagger X G)=0$. However, it holds $\swc(G^\dagger X G) = \swc(X)$, so we have a contradiction with Lemma~\ref{lemma-basic}.
    
     To prove the second statement, we consider two cases. If $\sw(A) = \pi$, then there is $\sigma \in \Sn{n}$ that $\tr(\sigma A)=0$. For any $\alpha$ we take $\rho = e^{i\frac{\alpha}{2}}\sigma$. Now, if $\sw(A) < \pi$, we have the decomposition $A=GDG^\dagger$. From Lemma~\ref{lemma-basic} without loss of generality, we can assume that the elements of $D$ are $e^{i \beta_{k}}$, where $\beta_{1}=0 \le \ldots \le \beta_{n} = \sw(A)$. Let us define two diagonal matrices $D_P, D_Q \psd$ for which all elements are zeros except $(D_P)_{1,1}, (D_P)_{n,n}, (D_Q)_{1,1}, (D_Q)_{n,n}$. We want to define them to satisfy 
    \begin{equation}
        1\times(e^{i \frac{\alpha}{2}} (D_P)_{1,1} + e^{-i \frac{\alpha}{2}} (D_Q)_{1,1}) + e^{i \sw(A)}\times(e^{i \frac{\alpha}{2}} (D_P)_{n,n} + e^{-i \frac{\alpha}{2}} (D_Q)_{n,n})=0.
    \end{equation}
    To do so, we can set $(D_P)_{1,1}=0$ and $(D_Q)_{1,1}=1$ and set $(D_P)_{n,n}, (D_Q)_{n,n} \ge 0$ to satisfy $e^{i \frac{\alpha}{2}} (D_P)_{n,n} + e^{-i \frac{\alpha}{2}} (D_Q)_{n,n} = e^{i(\pi - \sw(A) - \alpha/2)}$. The last is doable as $\pi - \sw(A) - \alpha/2 \in [-\alpha/2,\alpha/2]$. We define $\rho = e^{i \frac{\alpha}{2}}P+e^{-i \frac{\alpha}{2}}Q \in \Sn{n}(\alpha)$, where $P = c(G^{-1})^\dagger D_P G^{-1}$ and $Q = c(G^{-1})^\dagger D_Q G^{-1}$ and some normalization constant $c > 0$ for $\tr(P)+\tr(Q)=1$. Observe that $\tr(\rho A) = 0$.
\end{proof}
\begin{lemma}\label{lemma-equiv-exists}
Let $\A$ be a linear subspace of $\Mn{n}$ and $\alpha \in [0,\pi)$. Then, the following conditions are equivalent
\begin{enumerate}
    \item $\exists_{\rho \in \Sn{n}(\alpha)} \forall_{A \in \A} \tr(\rho A) = 0$,
    \item $\forall_{A \in \A} \exists_{\rho \in \Sn{n}(\alpha)} \tr(\rho A) = 0$.
\end{enumerate}
\end{lemma}
\begin{proof}
The former condition naturally implies the latter; hence we will show the reverse implication. Let $(A_l)_{l=1}^k \subset \A$ be a basis of $\A$. Define a set 
$B = \{(\tr(\rho A_l))_{l=1}^k: \rho \in \Sn{n}(\alpha)\}$. By Lemma~\ref{lem-prop-states} the set $B$ is convex and compact. If $0 \not\in B$, then by the hyperplane separation theorem there is a vector $\ket{v} \neq 0$, such that $\Re(\braket{v}{b}) < 0$ for any $b \in B$. This is a contradiction to the latter condition, so $0 \in B$. That finishes the proof.
\end{proof}

\begin{proof} (of Theorem~\ref{thm-tool})
The equivalence between conditions 2. and 3. is obvious. We will show the equivalence of 1. and 3..

If 3. holds, then by Lemma~\ref{lemma4} we have $\swc(\rho) + \sw(A) \ge \pi$ for any $A \in \A$. By Lemma~\ref{lem-prop-states} it implies $\alpha + \sw(A) \ge \pi$. Taking the infimum over $A \in \A$ we get $\alpha + \sws(\A) \ge \pi$.

If 1. holds, then $\sw(A) \ge \sws(\A) > 0$ for any $A \in \A$. In particular, we have $\alpha + \sw(A) \ge \pi$. By Lemma~\ref{lemma4} for any $A \in \A$ there is $\rho \in \Sn{n}(\alpha)$ such that $\tr(\rho A)=0$. By Lemma~\ref{lemma-equiv-exists} there is $\rho_0 \in \Sn{n}(\alpha)$ such that $\tr(\rho_0 A)=0$ for any $A \in \A$.

\end{proof}
\section{The proof of Theorem~\ref{thm-tensor}}\label{app-proof-thm-tensor}  

\begin{lemma}\label{lem-sw-prod}
For $A \in \Mn{n}$ and $B \in \Mn{m}$ we have $\sw(A \otimes B) = \min(\sw(A) + \sw(B), \pi)$.
\end{lemma}
\begin{proof}
    If $0 \in W(A)$ or $0 \in W(B)$, then $0 \in W(A \otimes B)$, so $\sw(A \otimes B) = \pi$ and $\sw(A) + \sw(B) \ge \pi$. Now, let us assume that $0 \not\in W(A)$ and $0 \not\in W(B)$. From Lemma~\ref{lemma-no-zero} we can write $A=G_AD_AG_A^\dagger$ and $B= G_BD_BG_B^\dagger$. From Lemma~\ref{lemma-basic} without loss of generality, we can assume that the elements of $D_A$ are $e^{i \alpha_{k}}$, where $\alpha_{1}=0 \le \ldots \le \alpha_{n} = \sw(A) < \pi$. The same is true for the elements $e^{i \beta_{l}}$ of $D_B$, where $\beta_{1}=0 \le \ldots \le \beta_{m} = \sw(B) < \pi$. Then, from Lemma~\ref{lemma-basic} we have $\sw(A\otimes B) = \sw(D_A \otimes D_B)$. If $\sw(A)+\sw(B) \le \pi$, then $\sw(D_A \otimes D_B) = \sw(A)+\sw(B)$. Otherwise, if $\sw(A)+\sw(B) > \pi$,  then $0 \in \conv\{1,e^{i \sw(A)}, e^{i (\sw(A)+\sw(B))}\}$, so $0 \in W(D_A \otimes D_B)$.
\end{proof}

\begin{proof} (of Theorem~\ref{thm-tensor})
    From Lemma~\ref{lem-sw-prod} we have
    \begin{equation}
      \sws(\A \otimes \B) \le \inf \{\sw(A \otimes B): A\in \A, B \in \B\} = \min(\sws(\A)+\sws(\B), \pi).
    \end{equation}
    To show the lower bound, we use the construction of $\rho_{\A \otimes \B}$ provided in the theorem statement.
    
     If $\sws(\A) = 0$ and $\sws(\B) = 0$, then $\sws(\A \otimes \B)=0$ and $\rho_{\A \otimes \B} = \Id_{nm}/(nm) = \rho_\A \otimes \rho_\B$. If $\sws(\A) > 0$ and $\sws(\B) = 0$ (similar for the reverse case), then we define $\rho = \rho_\A \otimes \rho_\B = e^{i \frac{\alpha_\A}{2}} P_\A \otimes \Id_m/m +  e^{-i \frac{\alpha_\A}{2}} Q_\A \otimes \Id_m/m$. Observe that $\rho \in \Sn{nm}(\alpha_\A)$ and $\tr(\rho (A\otimes B)) = 0$ for $A\in\A, B\in\B$. By Theorem~\ref{thm-tool} it means $\sws(\A \otimes \B) \ge \pi - \alpha_\A = \sws(\A)$. Therefore, $\sws(\A \otimes \B) = \sws(\A)$ and $\rho_{\A \otimes \B} = \rho$.

    For the next cases, let $\sws(\A) > 0$, $\sws(\B) > 0$. Combining Lemma~\ref{lem-prop-states} and Lemma~\ref{lemma4} we get $\swc(\rho_\A) = \alpha_\A$. In particular, $P_\A, Q_\A \neq 0$.
    
     If $\sws(\A) + \sws(\B) < \pi$ we take $\rho \propto e^{i\frac{\alpha_\A + \alpha_\B - \pi}{2}} P_\A \otimes P_\B + e^{-i\frac{\alpha_\A + \alpha_\B - \pi}{2}} Q_\A \otimes Q_\B$. Observe that $\rho \in \Sn{nm}(\alpha_\A + \alpha_\B - \pi)=\Sn{nm}(\pi - (\sws(\A)+\sws(\B)))$ as $\tr(P_\A)\tr(P_\B)+\tr(Q_\A)\tr(Q_\B) > 0$. By the assumptions
    \begin{equation}
        \begin{split}
            0&=\tr(\rho_\A A) = e^{i \frac{\alpha_\A}{2}} \tr(P_\A A) +  e^{-i \frac{\alpha_\A}{2}}\tr(Q_\A A),\\
            0&=\tr(\rho_\B B) = e^{i \frac{\alpha_\B}{2}} \tr(P_\B B) +  e^{-i \frac{\alpha_\B}{2}}\tr(Q_\B B),
        \end{split}
    \end{equation}
    where $A \in \A$ and $B \in \B$. From that we get $\tr(\rho (A \otimes B)) = 0$. By Theorem~\ref{thm-tool} it means $\sws(\A \otimes \B) \ge \pi - (\pi - (\sws(\A)+\sws(\B))) = \sws(\A)+\sws(\B)$. Therefore, $\sws(\A \otimes \B) = \sws(\A) + \sws(\B)$ and $\rho_{\A \otimes \B} = \rho$.

    If $\sws(\A) + \sws(\B) \ge \pi$ we take $\rho  \propto \rho_\A \otimes \rho_\B^\dagger + \rho_\A^\dagger \otimes \rho_\B$. Observe that $\rho=\rho^\dagger$. Moreover, $\Re W(\rho_\A \otimes \rho_\B^\dagger) \ge 0$, so by Lemma~\ref{lemma-basic} we have $\rho \psd$.
 Additionally, $\Re \tr(\rho_\A \otimes \rho_\B^\dagger) > 0$, so $\tr(\rho) > 0$. In summary, we have $\rho \in \Sn{nm}(0)$. By assumptions $\tr(\rho_\A A)=\tr(\rho_\B B)=0$ for $A \in \A$ and $B \in \B$. From that we get $\tr(\rho (A \otimes B)) = 0$. By Theorem~\ref{thm-tool} it means $\sws(\A \otimes \B) = \pi$ $\rho_{\A \otimes \B} = \rho$.
\end{proof}
 
\section{The proof of Corollary~\ref{cor-conj}}\label{app-proof-cor-conj}  
\begin{proof}
    The main tool in this proof is to evaluate the value of $\sws(S_\beta)$, where $S_\beta = \vspan(D_1, D_2)$ and $D_1 = \proj{0}+e^{i \beta} \proj{1}, D_2 = \proj{1}+e^{i \beta} \proj{2}$ for $\beta \in [\pi/2, \pi]$.

    For $\beta > \pi/2$ we define $\rho \in \Mn{3}$ given by
    \begin{equation}
        \rho \propto e^{i(\pi-\beta)} \diag(0,1,2 \cos(\pi - \beta))+e^{-i(\pi-\beta)} \diag(2 \cos(\pi - \beta),1, 0),
    \end{equation}
    where $\diag$ represents the diagonal matrix. Observe that $\rho \in \Sn{3}(2(\pi-\beta))$ and $\tr(\rho D_1) = \tr(\rho D_2) = 0$. By Theorem~\ref{thm-tool} we get $\sws(S_\beta) \ge 2\beta - \pi$.

    For $\beta < \pi$ we define $D = D_1 + e^{i\omega}D_2 = \diag(1, e^{i\beta}+e^{i\omega}, e^{i(\beta+\omega)})$, where $\omega = \beta+\pi+\epsilon$. Let us analyse $W(D)$ for $\epsilon \to_+ 0$. This numerical range is a triangle with vertices given by: $1, e^{i\beta}+e^{i\omega},e^{i(\beta+\omega)}$. Their arguments are, respectively, $0,\beta - \pi/2 + \epsilon'$ and $2\beta - \pi + \epsilon$, where $\epsilon' \to_+ 0$ with $\epsilon \to_+ 0$. It implies $\sw(D) = 2\beta - \pi + \epsilon$, and hence $\sws(S_\beta) \le 2\beta - \pi$.

    We have shown that $\sws(S_\beta) = 2\beta - \pi$ for $\beta \in [\pi/2, \pi]$. By Theorem~\ref{thm-tool} there is a
density operator in the orthogonal complement of $(S_\beta)^{\otimes N}$ if and only if $\sws((S_\beta)^{\otimes N})\ge \pi$. By Corollary~\ref{cor-many} it holds if and only if $N \sws(S_\beta) \ge \pi$, which is equivalent to $\beta \in [\pi/2 + \pi/(2N), \pi]$.

\end{proof}

\section{The proof of Theorem~\ref{thm-discr}}\label{app-proof-thm-discr}  
\begin{proof}
    The channel $\Phi$ and $\Psi$ are perfectly discriminable in the parallel scheme if and only if there exists $\rho \in \Sn{d^N}$ for some number $N$ of copies such that $\tr(\rho M) = 0$ for any $M \in \A^{\otimes N}$. By Theorem~\ref{thm-tool}, it is equivalent to condition $\sws(\A^{\otimes N})=\pi$ for some $N$. Using Corollary~\ref{cor-many} it follows that $N \sws(\A) \ge \pi$. For $\sws(\A) > 0$ the minimal number of copies has to satisfy this inequality.
\end{proof}

\bibliographystyle{ieeetr}
\bibliography{thebibliography}

\begin{thebibliography}{10}

\bibitem{duan2016parallel}
R.~Duan, C.~Guo, C.-K. Li, and Y.~Li, ``Parallel distinguishability of quantum
  operations,'' in {\em 2016 IEEE International Symposium on Information Theory
  (ISIT)}, pp.~2259--2263, IEEE, 2016.

\bibitem{helstrom1969quantum}
C.~W. Helstrom, ``Quantum detection and estimation theory,'' {\em Journal of
  Statistical Physics}, vol.~1, no.~2, pp.~231--252, 1969.

\bibitem{holevo1973bounds}
A.~S. Holevo, ``Bounds for the quantity of information transmitted by a quantum
  communication channel,'' {\em Problemy Peredachi Informatsii}, vol.~9, no.~3,
  pp.~3--11, 1973.
\newblock English translation: Problems of Information Transmission, vol.~9,
  no.~3, pp.~177--183, 1973.

\bibitem{acin2001statistical}
A.~Ac{\'\i}n, ``Statistical distinguishability between unitary operations,''
  {\em Physical Review Letters}, vol.~87, no.~17, p.~177901, 2001.

\bibitem{duan2007entanglement}
R.~Duan, Y.~Feng, and M.~Ying, ``Entanglement is not necessary for perfect
  discrimination between unitary operations,'' {\em Physical Review Letters},
  vol.~98, no.~10, p.~100503, 2007.

\bibitem{duan2008local}
R.~Duan, Y.~Feng, and M.~Ying, ``Local distinguishability of multipartite
  unitary operations,'' {\em Physical Review Letters}, vol.~100, no.~2,
  p.~020503, 2008.

\bibitem{puchala2018strategies}
Z.~Pucha{\l}a, {\L}.~Pawela, A.~Krawiec, and R.~Kukulski, ``Strategies for
  optimal single-shot discrimination of quantum measurements,'' {\em Physical
  Review A}, vol.~98, no.~4, p.~042103, 2018.

\bibitem{krawiec2020discrimination}
A.~Krawiec, {\L}.~Pawela, and Z.~Pucha{\l}a, ``Discrimination of {POVMs} with
  rank-one effects,'' {\em Quantum Information Processing}, vol.~19, no.~12,
  p.~428, 2020.

\bibitem{lewandowska2023strategies}
P.~Lewandowska, {\L}.~Pawela, and Z.~Pucha{\l}a, ``Strategies for single-shot
  discrimination of process matrices,'' {\em Scientific Reports}, vol.~13,
  no.~1, p.~3046, 2023.

\bibitem{chiribella2013identification}
G.~Chiribella, G.~M. D'Ariano, and M.~Roetteler, ``Identification of a
  reversible quantum gate: assessing the resources,'' {\em New Journal of
  Physics}, vol.~15, no.~10, p.~103019, 2013.

\bibitem{puchala2021multiple}
Z.~Pucha{\l}a, {\L}.~Pawela, A.~Krawiec, R.~Kukulski, and M.~Oszmaniec,
  ``Multiple-shot and unambiguous discrimination of von {N}eumann
  measurements,'' {\em Quantum}, vol.~5, p.~425, 2021.

\bibitem{duan2009perfect}
R.~Duan, Y.~Feng, and M.~Ying, ``Perfect distinguishability of quantum
  operations,'' {\em Physical Review Letters}, vol.~103, no.~21, p.~210501,
  2009.

\bibitem{watrous2018theory}
J.~Watrous, {\em The Theory of Quantum Information}.
\newblock Cambridge: Cambridge University Press, 2018.

\bibitem{harrow2010adaptive}
A.~W. Harrow, A.~Hassidim, D.~W. Leung, and J.~Watrous, ``Adaptive versus
  nonadaptive strategies for quantum channel discrimination,'' {\em Physical
  Review A}, vol.~81, no.~3, p.~032339, 2010.

\bibitem{cope2018adaptive}
T.~P.~W. Cope and S.~Pirandola, ``Adaptive estimation and discrimination of
  {H}olevo--{W}erner channels,'' {\em Quantum Measurements and Quantum
  Metrology}, vol.~4, no.~1, 2017.
\newblock arXiv:1801.05441.

\bibitem{pirandola2017ultimate}
S.~Pirandola and C.~Lupo, ``Ultimate precision of adaptive noise estimation,''
  {\em Physical Review Letters}, vol.~118, no.~10, p.~100502, 2017.

\bibitem{liu2019distinguishing}
S.~Liu, Y.~Li, and R.~Duan, ``Distinguishing unitary gates on the {IBM} quantum
  processor,'' {\em Science China Information Sciences}, vol.~62, no.~7,
  p.~72502, 2019.

\bibitem{debry2023experimental}
K.~DeBry, J.~Sinanan-Singh, C.~D. Bruzewicz, D.~Reens, M.~E. Kim, M.~P.
  Roychowdhury, R.~McConnell, I.~L. Chuang, and J.~Chiaverini, ``Experimental
  quantum channel discrimination using metastable states of a trapped ion,''
  {\em Physical Review Letters}, vol.~131, no.~17, p.~170602, 2023.

\bibitem{bilek2026experimental}
A.~B{\'\i}lek, J.~Hlisnikovsk{\'y}, T.~Bezd{\v{e}}k, R.~Kukulski, and
  P.~Lewandowska, ``Experimental study of multiple-shot unitary channels
  discrimination using the {IBM} {Q} computers,'' {\em Scientific Reports},
  vol.~16, no.~1, p.~6142, 2026.

\bibitem{li2019note}
C.-K. Li, Y.~Liu, C.~Ma, and D.~C.~P. Pelejo, ``A note on parallel
  distinguishability of two quantum operations,'' {\em Electronic Journal of
  Linear Algebra}, vol.~36, pp.~198--209, 2020.

\bibitem{vandenberghe1996sdp}
L.~Vandenberghe and S.~Boyd, ``Semidefinite programming,'' {\em SIAM Review},
  vol.~38, no.~1, pp.~49--95, 1996.

\bibitem{numericalshadow}
{\L}.~Pawela, P.~Gawron, J.~A. Miszczak, Z.~Pucha{\l}a, K.~{\.Z}yczkowski,
  P.~Lewandowska, and R.~Kukulski, ``Numerical shadow.''
  \url{https://numericalshadow.org/}.
\newblock Institute of Theoretical and Applied Informatics, Polish Academy of
  Sciences. Accessed: 2026-09-15.

\bibitem{Puchala2011}
Z.~Pucha{\l}a, P.~Gawron, J.~A. Miszczak, {\L}.~Skowronek, M.-D. Choi, and
  K.~{\.Z}yczkowski, ``Product numerical range in a space with tensor product
  structure,'' {\em Linear Algebra and its Applications}, vol.~434, no.~1,
  pp.~327--342, 2011.

\bibitem{Gawron2010}
P.~Gawron, Z.~Pucha{\l}a, J.~A. Miszczak, {\L}.~Skowronek, and
  K.~{\.Z}yczkowski, ``Restricted numerical range: A versatile tool in the
  theory of quantum information,'' {\em Journal of Mathematical Physics},
  vol.~51, no.~10, p.~102204, 2010.

\bibitem{li2020joint}
C.-K. Li, Y.-T. Poon, and Y.-S. Wang, ``Joint numerical ranges and
  commutativity of matrices,'' {\em Journal of Mathematical Analysis and
  Applications}, vol.~491, no.~1, p.~124310, 2020.

\bibitem{plaumann2021kippenhahn}
D.~Plaumann, R.~Sinn, and S.~Weis, ``{K}ippenhahn's theorem for joint numerical
  ranges and quantum states,'' {\em SIAM Journal on Applied Algebra and
  Geometry}, vol.~5, no.~1, pp.~86--113, 2021.

\bibitem{keeler1997numerical}
D.~S. Keeler, L.~Rodman, and I.~M. Spitkovsky, ``The numerical range of
  $3\times 3$ matrices,'' {\em Linear Algebra and its Applications}, vol.~252,
  no.~1--3, pp.~115--139, 1997.

\bibitem{szymanski2018classification}
K.~Szyma{\'n}ski, S.~Weis, and K.~{\.Z}yczkowski, ``Classification of joint
  numerical ranges of three {H}ermitian matrices of size three,'' {\em Linear
  Algebra and its Applications}, vol.~545, pp.~148--173, 2018.

\bibitem{horn2006canonical}
R.~A. Horn and V.~V. Sergeichuk, ``Canonical forms for complex matrix
  congruence and {$*$}congruence,'' {\em Linear Algebra and its Applications},
  vol.~416, no.~2--3, pp.~1010--1032, 2006.

\end{thebibliography}

\end{document}